\pdfoutput=1
\RequirePackage{ifpdf}
\ifpdf 
\documentclass[pdftex]{sigma}
\else
\documentclass{sigma}
\fi

\numberwithin{equation}{section}

\newtheorem{Theorem}{Theorem}[section]
\newtheorem*{Theorem*}{Theorem}

\newtheorem{Proposition}[Theorem]{Proposition}
 { \theoremstyle{definition}

\newtheorem{Remark}[Theorem]{Remark} }

\begin{document}


\newcommand{\arXivNumber}{2401.10445}

\renewcommand{\PaperNumber}{081}
	
\FirstPageHeading

\ShortArticleName{Multidimensional Nonhomogeneous Quasi-Linear Systems}

\ArticleName{Multidimensional Nonhomogeneous Quasi-Linear\\ Systems and Their Hamiltonian Structure}

\Author{Xin HU and Matteo CASATI}
\AuthorNameForHeading{X.~Hu and M.~Casati}

\Address{School of Mathematics and Statistics, Ningbo University, Ningbo 315211, P.R.~China}
\Email{\href{mailto:1826490374@qq.com}{1826490374@qq.com}, \href{mailto:matteo@nbu.edu.cn}{matteo@nbu.edu.cn}}
\URLaddress{\url{https://math.nbu.edu.cn/info/1046/4832.htm}}

\ArticleDates{Received July 09, 2024, in final form September 04, 2024; Published online September 10, 2024}

\Abstract{In this paper, we investigate multidimensional first-order quasi-linear systems and find necessary conditions for them to admit Hamiltonian formulation. The insufficiency of the conditions is related to the Poisson cohomology of the admissible Hamiltonian operators. We present in detail the examples of two-dimensional, two-components systems of hydrodynamic type and of a real reduction of the 3-waves system.}


\Classification{37K10; 37K25}

\section{Introduction}
Hamiltonian systems are mathematical models used to describe the dynamical behavior of physical systems. They are based on Hamiltonian principle and Hamiltonian equations, which utilize generalized coordinates and generalized momenta to characterize the state of the system. Hamiltonian systems find wide-ranging applications in classical mechanics, quantum mechanics, and other fields. An autonomous system of evolutionary PDEs
\begin{gather*}
u^i_t=F^i\big(u^j,u^j_x,\dots,u^j_{kx}\big), \qquad i,j=1,2,\dots,n
\end{gather*}
in two dependent variables $(x,t)$ is referred to as Hamiltonian, with a Hamiltonian structure~$P$ and a Hamiltonian functional $H$, if it can be
expressed as
$
u^i_t=P^{ij}\frac{\delta H}{\delta u^j}$,
where $P^{ij}$ is a~skewadjoint (pseudo)-differential operator with vanishing Schouten torsion (see \cite{olver1993applications}), referred to as a~Hamiltonian structure, and $H=\int h\big(u^j,u^j_x,\dots\big) {\rm d}x$ is a local functional.

To date, many researchers have applied advanced mathematical methods and techniques, such as Lie algebra \cite{olver1993applications}, Poisson geometry \cite{weinstein1983local}, symplectic geometry \cite{mokhov1995symplectic}, etc., to investigate the Hamiltonian structures of nonlinear evolutionary equations, including the KdV equation, the KN equation \cite{mokhov1995symplectic}, the sine-Gordon equation, the KP equation \cite{dubrovin2001integrable,olver1980hamiltonian}, and so on.

In the paper by Vergallo \cite{vergallo2023non}, an investigation has been conducted on a class of evolutionary systems known as first-order quasi-linear systems (it was firstly investigated in \cite{Vergallo_2022})
\begin{gather*}
u^i_t=V^i_l(u)u^l_x+W^i(u), \qquad i=1,2,\dots, n.
\end{gather*}
Through the application of cotangent covering theory \cite{krasil1989nonlocal,vergallo2021homogeneous}, the author conducted a study on the compatibility conditions for (1+0)-order nonhomogeneous quasi-linear systems to be Hamiltonian systems under the nonhomogeneous hydrodynamic operator.

A multidimensional Hamiltonian system is an evolutionary system
\begin{gather*}
u^i_t=F^i\big(u^j,u^j_\alpha,u^j_{\alpha\beta},\dots\big), \qquad i,j=1,2,\dots,n,
\end{gather*}
which can be written as \smash{$u^i_t=P^{ij}\frac{\delta H}{\delta u^j}$}, where $P^{ij}$ is a Hamiltonian operator and $H=\int h\big(u^j,u^j_\alpha,\allowbreak \smash{u^j_{\alpha\beta},\dots\big) {\rm d}x^D}$ and $\alpha ,\beta ,\dots$ is a component of variable $x$ in $D$-dimensional space. \smash{$u^j_\alpha$} is a shorthand notation for \smash{$\frac{\partial u^j}{\partial x^\alpha}$}, \smash{$u^j_{\alpha\beta}$} for \smash{$\frac{\partial^2u^j}{\partial x^\alpha \partial x^\beta}$}, etc.

Therefore, the idea of this paper is to extend the study of quasi-linear systems done by Vergallo to include the $D$-dimensional case ($D>1$), namely
\begin{equation}\label{multiquasi}
u^i_t=V^{i\alpha}_l(u)u^l_\alpha+W^i(u), \qquad i=1,2,\dots, n.
\end{equation}
The question is: which quasi-linear systems admit Hamiltonian formulations (namely the systems \eqref{multiquasi} can be expressed as $u^i_t=P^{ij}\frac{\delta H}{\delta u^j}$) in the multidimensional case? The main approach involves the application of Poisson vertex algebra (PVA) theory to compute the compatibility conditions of $D$-dimensional nonhomogeneous Hamiltonian operators, along with the application of cotangent covering theory to find compatible quasi-linear systems under such operators. Finally, the compatibility conditions for quasi-linear systems to be Hamiltonian are computed in the $D=N=2$ case, and we relate the compatibility conditions with the first Poisson cohomology group of the Hamiltonian operator.
\section{Nonhomogeneous multidimensional Hamiltonian structures}
\subsection{PVAs and Hamiltonian structures}
Let us consider the space of maps $\mathcal{M}$ from a $D$-dimensional manifold $X$ to a $N$-dimensional target manifold $U$. Chosen the coordinate systems ${u^i}$ and ${x^\alpha}$ , the quotient space of local functionals is defined as $\mathcal{F}:=\frac{\mathcal{A}}{\partial_1\mathcal{A}+\partial_2\mathcal{A}+\dots+\partial_D\mathcal{A}}$ \cite{olver1993applications}, where $\mathcal{A}=\mathcal{A}(U)$ is the space of differential polynomials. Now, we endow the space $\mathcal{F}$ with a Lie bracket $\lbrace \ , \ \rbrace$, also known as the local Poisson bracket.
A Lie bracket on $\mathcal F$ can equivalently defined in terms of a local Poisson bivector, $P \in \Lambda^2$, by
\begin{equation}\label{eq:pb-def}
P(\delta F,\delta G)=\lbrace F,G \rbrace =\int \frac{\delta F}{\delta u^i}P^{ij}\frac{\delta G}{\delta u^j},
\end{equation}
where \smash{$P^{ij}=P^{ij}_S\partial^S$}, $P^{ij}_S \in \mathcal{A}$, is a skewadjoint differential operator such that $[P,P]=0$, referred to as a Hamiltonian structure, here we denote
\begin{gather*}
S=(s_1,s_2,\dots,s_D), \qquad \partial^S=\left(\frac{{\rm d}}{{\rm d}x^1}\right)^{s_1}\left(\frac{\rm d}{{\rm d}x^2}\right)^{s_2}\cdots\left(\frac{\rm d}{{\rm d}x^D}\right)^{s_D}, \\ \frac{\delta F}{\delta u^i}=\sum\limits_{S}(-1)^{s_1+s_2+\dots+s_D}\partial^S\left(\frac{\partial F}{\partial\big(\partial^Su^i\big)}\right).
\end{gather*}
 The skewadjointness of $P$ is equivalent to the skewsymmetry of the Poisson bracket and the vanishing of the Schouten torsion $[P,P]$ \cite{olver1993applications} corresponds to the Jacobi identity.

Next, we will briefly demonstrate the application of Poisson vertex algebra (PVA) theory to compute the conditions for an operator to be Hamiltonian. Detailed information about PVA can be found in \cite{barakat2009poisson}. A multidimensional Poisson vertex algebra (mPVA) is a differential algebra~$(\mathcal{V}, \partial)$ endowed with
a bilinear operation $\mathcal{V} \times \mathcal{V}\rightarrow \mathbb{R}[\lambda_1,\dots,\lambda_D] \otimes\mathcal{V}$ called the $\lambda$-bracket and denoted $\lbrace f_{\lambda} g\rbrace=C_{s_1,s_2,\dots,s_D}(f,g)\lambda^{s_1}_1\lambda^{s_2}_2\cdots\lambda^{s_D}_D=C_S(f,g)\lambda^S$. We take $\mathcal{V}$ to be the space $\mathcal{A}$ of differential polynomials and \smash{$\partial_\alpha=u^i_{S+\xi_\alpha}\frac{\partial}{\partial u^i_S}$}, where $u^i_S$ is a shorthand notation for $\partial^S u^i$ and
 \[
\xi_{\alpha}=(0,\ldots,0,\underbrace{1}_{\alpha},0,\ldots)
\] is the canonical basic vectors in $\mathbb{Z}^D$.

We can explicitly compute a $\lambda$-bracket between to differential polynomials using the so-called master formula for a multidimensional Poisson vertex algebra \cite{casati2015deformations} as it follows
\begin{equation}\label{eq:master}
\lbrace f_{\lambda}g\rbrace=\sum_{i,j=1,\dots,N}\sum_{S,M\in \mathbb{Z}^D_{\geq0}}\frac{\partial g}{\partial u^j_S}(\lambda+\partial)^S\big\lbrace u^i_{\lambda+\partial}u^j\big\rbrace(-\lambda-\partial)^M\frac{\partial f}{\partial u^i_M}.
\end{equation}

The following two theorems \cite{casati2015deformations} are presented to illustrate the relationship between Poisson vertex algebras and local Poisson brackets.
\begin{Theorem}
Let us define a bilinear bracket on $\mathcal{F}$
\[
\lbrace F,G\rbrace:=\int {\lbrace f_{\lambda} g\rbrace}|_{\lambda=0},
\]
where $f,g\in \mathcal{A}$ are the densities of $F$, $G$.
If the $\lambda$-bracket satisfies the axioms of a PVA, then the bracket we defined is a local Poisson bracket.
\end{Theorem}
\begin{Theorem}
Given a Hamiltonian structure $P=P^{ij}_S\partial^S$ on $\mathcal{A}$,
the $\lambda$-bracket defined on generators as
\[
\big\lbrace u^i_{\lambda}u^j\big\rbrace:=\sum P^{ji}_S\lambda^S
\]
and extended to the full algebra using the master formula satisfies the axioms of the mPVA.
\end{Theorem}
These two theorems demonstrate the equivalence between mPVAs and multidimensional Hamiltonian structures. In particular, the Jacobi identity is equivalent to the so-called PVA-Jacobi identity. Thus, we can obtain a Hamiltonian structure by computing the Jacobi identity of the $D$-dimensional $\lambda$-bracket, among the generators (it is equivalent to the Jacobi identity of the generic densities), namely
\begin{equation}\label{PJ}
\big\lbrace u^i_{\lambda}\big\lbrace u^j_{\mu} u^k\big\rbrace\big\rbrace -\big\lbrace u^j_{\mu}\big\lbrace u^i_{\lambda} u^k\big\rbrace\big\rbrace =\big\lbrace\big\lbrace u^i_{\lambda} u^j\big\rbrace_{\mu +\lambda}u^k\big\rbrace.
\end{equation}

\subsection{The homogeneous case}
Let us consider homogeneous differential operators of order $m$ with the following form:
\begin{gather*}
P^{ij}=g^{ij}_{S_1}(u)\partial^{S_1}+b^{ij\alpha}_{kS_2}(u) u^k_\alpha\partial^{S_2}+c^{ij\alpha\beta}_{kS_3}(u)u^k_{\alpha\beta}\partial^{S_3}+c^{ij\alpha\beta}_{klS_3}(u)u^k_\alpha u^l_\beta \partial^{S_3}+\cdots,
\end{gather*}
where
\begin{gather*}
S_n=(s_1,s_2,\dots,s_D) , \qquad n=1,2,\dots,
\\
\partial ^{S_n}=\left(\frac{\rm d}{{\rm d}x^1}\right)^{s_1}\left(\frac{\rm d}{{\rm d}x^2}\right)^{s_2}\cdots \left(\frac{\rm d}{{\rm d}x^D}\right)^{s_D} , \qquad s_1+s_2+\dots+s_D=m-n+1.
\end{gather*}
One important class of Poisson brackets is the Poisson bracket for which a first-order homogeneous differential operator
\begin{equation}\label{OP}
P^{ij}=g^{ij\alpha}(u)\partial_\alpha +b^{ij\alpha}_k(u)u^k_\alpha
 \end{equation}
satisfies $P^* =-P$ and $[P,P]=0$. Such brackets, defined by Dubrovin and Novikov \cite{dubrovin1996hamiltonian}, are called Poisson brackets of hydrodynamic type or DN brackets.

We recall the conditions satisfied by the coefficients $\big(g^{ij\alpha},b^{ij\alpha}_k\big)$ in the multidimensional DN brackets, provided by Mokhov~\cite{mokhov1988poisson}.

\begin{Theorem}
The operator \eqref{OP} is Hamiltonian if and only if
\begin{gather}
 g^{ij\alpha}=g^{ji\alpha},\qquad
 \frac{\partial g^{ij\alpha}}{\partial u^k}=b^{ij\alpha}_k+b^{ji\alpha}_k,\qquad
 \sum_{(\alpha,\beta)}\big(g^{li\alpha}b^{jk\beta}_l-g^{lj\beta}b^{ik\alpha}_l\big)=0,\nonumber\\
 \sum_{(i,j,k)}\big(g^{li\alpha}b^{jk\beta}_l-g^{lj\beta}b^{ik\alpha}_l\big)=0,\qquad
 \sum_{(\alpha,\beta)}\left[g^{li\alpha}\left(\frac{\partial b^{jk\beta}_l}{\partial u^r}-\frac{\partial b^{jk\beta}_r}{\partial u^l} \right)+b^{ij\alpha}_lb^{lk\beta}_r-b^{ik\alpha}_lb^{lj\beta}_r\right]=0,\nonumber\\
 g^{li\beta}\frac{\partial b^{jk\alpha}_r}{\partial u^l}-b^{ij\beta}_lb^{lk\alpha}_r-b^{ik\beta}_lb^{jl\alpha}_r=g^{lj\alpha}\frac{\partial b^{ik\beta}_r}{\partial u^l}-b^{ji\alpha}_lb^{lk\beta}_r-b^{jk\alpha}_lb^{il\beta}_r,\nonumber\\
 \frac{\partial}{\partial u^s}\left[g^{li\alpha}\left(\frac{\partial b^{jk\beta}_l}{\partial u^r}-\frac{\partial b^{jk\beta}_r}{\partial u^l}\right)+b^{ij\alpha}_lb^{lk\beta}_r-b^{ik\alpha}_lb^{lj\beta}_r\right] \nonumber \\
 \qquad+\frac{\partial}{\partial u^r}\left[g^{li\beta}\left(\frac{\partial b^{jk\alpha}_l}{\partial u^s}-\frac{\partial b^{jk\alpha}_s}{\partial u^l}\right)+b^{ij\beta}_lb^{lk\alpha}_s-b^{ik\beta}_lb^{lj\alpha}_s\right] \nonumber \\
\qquad +\sum_{(i,j,k)}\left[b^{li\beta}_r\left(\frac{\partial b^{jk\alpha}_s}{\partial u^l}-\frac{\partial b^{jk\alpha}_l}{\partial u^s}\right)\right] +\sum_{(i,j,k)}\left[b^{li\alpha}_s\left(\frac{\partial b^{jk\beta}_r}{\partial u^l}-\frac{\partial b^{jk\beta}_l}{\partial u^r}\right)\right]=0,\label{con1}
\end{gather}
where $\sum$ represents the cyclic summation over indices.
\end{Theorem}

In analogy to the $D=1$ case, the coefficients $g^{ij\alpha}$ are components of a symmetric (2,0)-tensor. Assuming that $g^\alpha$ is non-degenerate and defining $b^{ij\alpha}_k=-g^{il\alpha}\Gamma^{j\alpha}_{lk}$, we observe that $\Gamma^\alpha$ are Christoffels symbols of the Levi-Civita connection of $g^\alpha$.

\subsection {The nonhomogeneous case}
Following the content of the previous section, we extend the first-order homogeneous Hamiltonian operator \eqref{OP} to the nonhomogeneous operator of type 1+0
\begin{equation}\label{nonOP}
P^{ij}=g^{ij\alpha}(u)\partial_\alpha+b^{ij\alpha}_k(u)u^k_\alpha+\omega^{ij}(u).
\end{equation}
We need to find the conditions for which the operator \eqref{nonOP} is Hamiltonian, namely the brack\-et~${\lbrace F,G\rbrace=\int \frac{\delta F}{\delta u^i}P^{ij}\frac{\delta G}{\delta u^j}}$ satisfies the property of skewsymmetry and Jacobi identity.

Firstly, let us recall the condition for the 0-order ultralocal operator $\omega^{ij}(u)$ to be Hamiltonian.
\begin{Theorem}
The operator $\omega^{ij}(u)$ is Hamiltonian if and only if
\begin{gather}
\omega^{ij}=-\omega^{ji},\qquad
\omega^{il}\omega^{jk}_{,l}+\omega^{jl}\omega^{ki}_{,l}+\omega^{kl}\omega^{ij}_{,l}=0,\label{con2}
\end{gather}
where $\omega^{ij}_{,l}$ is a shorthand notation for $\frac{\partial \omega^{ij}}{\partial u^l}$.
\end{Theorem}
\begin{Remark}
The condition is the same as for a Poisson bivector on a finite-dimensional mani\-fold.
\end{Remark}

In the previous section, we have already stated the Hamiltonian condition for a first-order homogeneous operator. For the nonhomogeneous operator \eqref{nonOP}, the Hamiltonian condition is the following:
\begin{Theorem}[\cite{mokhov1995symplectic}]\label{theor5}
The operator \eqref{nonOP} is Hamiltonian if and only if the operator $g^{ij\alpha}\partial_\alpha+b^{ij\alpha}_ku^k_\alpha$ satisfies the conditions \eqref{con1}, $\omega^{ij}$ satisfies the conditions \eqref{con2}, and the following additional conditions hold true:
\begin{subequations}\label{con3}
\begin{align}
&T^{ikj\alpha} =T^{kji\alpha},\label{con3-1} \\
&T_{,s}^{kji\alpha} =\sum_{(i,j,k)}b_{s}^{li\alpha}\omega_{,l}^{kj}+\big(b_{s,l}^{ik\alpha}-b_{l,s}^{ik\alpha}\big)\omega^{lj}\label{con3-2},
\end{align}
\end{subequations}
where $T^{ikj\alpha}=g^{il\alpha}\omega^{kj}_{,l}-b^{ik\alpha}_l\omega^{lj}-b^{ij\alpha}_l\omega^{kl}$.
\end{Theorem}
\begin{proof}
Our proof is a direct computation. The $\lambda$-bracket corresponding to $P$ is $\big\lbrace u^i_{\lambda}u^j\big\rbrace=\smash{P^{ji}(\lambda)=g^{ji\alpha}\lambda_\alpha+b^{ji\alpha}_ku^k_\alpha+\omega^{ji}}$ by applying the theory of mPVA. We compute the PVA-Jacobi identity of the generators:
\begin{subequations}
\begin{gather}
\big\lbrace u^i_{\lambda}\big\lbrace u^j_{\mu} u^k\big\rbrace\big\rbrace=\sum\frac{\partial\big(g^{kj\alpha}\mu_\alpha+b^{kj\alpha}_hu^h_\alpha+\omega^{kj}\big)}{\partial u^l_{(M)}}(\lambda+\partial)^M\big(g^{li\beta}\lambda_\beta+b^{li\beta}_su^s_\beta+\omega^{li}\big),\label{Ja1}\\
\big\lbrace u^j_{\mu}\big\lbrace u^i_{\lambda} u^k\big\rbrace\big\rbrace=\sum\frac{\partial\big(g^{ki\alpha}\lambda_\alpha+b^{ki\alpha}_hu^h_\alpha+\omega^{ki}\big)}{\partial u^l_{(M)}}(\mu+\partial)^M\big(g^{lj\beta}\mu_\beta+b^{lj\beta}_su^s_\beta+\omega^{lj}\big),\label{Ja2}\\
\big\lbrace\big\lbrace u^i_{\lambda} u^j\big\rbrace_{\mu+\lambda}u^k\big\rbrace=\sum\big[g^{kl\beta} \big(\lambda_\beta+\mu_\beta+\partial_\beta\big)+b^{kl\beta}_hu^h_\beta+\omega^{kl}\big]\nonumber\\
\phantom{\big\lbrace\big\lbrace u^i_{\lambda} u^j\big\rbrace_{\mu+\lambda}u^k\big\rbrace=}{}\times(-\lambda-\mu-\partial)^M\frac{\partial\big(g^{ji\alpha} \lambda_\alpha+b^{ji\alpha}_su^s_\alpha+\omega^{ji}\big)}{\partial u^l_{(M)}},\label{Ja3}
\end{gather}
\end{subequations}
where the summation range is $1\leq h,s,l\leq N$, $1\leq\alpha,\beta\leq D$, and $M\in\mathbb{Z}^D_{\geq0}$. The equation $\eqref{Ja1}-\eqref{Ja2}=\eqref{Ja3}$ is PVA-Jacobi \eqref{PJ}, extracting the term $\lambda_\alpha$ leads to \eqref{con3-1}, and extracting the term $u^s_\alpha$ leads to \eqref{con3-2}. The remaining terms satisfy the Jacobi identity for the first-order homogeneous operator and the 0-order operator.
\end{proof}

\begin{Remark}
The conditions \eqref{con3-1} and \eqref{con3-2} are independent sets of equations for each spatial direction and are immediately generalized for any $D>0$; in particular, in the one-dimensional case and for nondegenerate $g$ they are equivalent to the local case of Ferapontov and Mokhov's results \cite{mokhov1994hamiltonian}; the local, multidimensional case was first investigated by Mokhov \cite{mokhov1995symplectic}.
\end{Remark}

\section[Quasi-linear systems and compatibility in the multidimensional case]{Quasi-linear systems and compatibility\\ in the multidimensional case}
\subsection{Quasi-linear systems}
In this section, we focus on $(D+1)$-dimensional first-order quasilinear systems $(D>1)$
\begin{equation}\label{sys}
u^i_t=V^{i\alpha}_l(u)u^l_\alpha+W^i(u), \qquad i=1,2,\dots, n.
\end{equation}
The systems \eqref{sys} (also referred to as the type 1+0) are the sum of 1-order homogeneous systems and 0-order systems. As an example, consider the $N$-waves system \cite{ablowitz1975nonlinear} in two dimensions:
\begin{gather}
u_{ij,t}=\alpha_{ij}u_{ij,x}+\beta_{ij}u_{ij,y}+\sum_{k\neq i,j}(\alpha_{ik}-\alpha_{kj})u_{ik}u_{kj},\qquad
\text{for}\quad i\neq j\quad \text{and}\quad u_{ii}=0, \nonumber \\
 i,j=1,\dots,N,\label{N-W}
\end{gather}
where $\alpha_{ij}=\alpha_{ji}\in\mathbb{R}$, $\beta_{ij}=\beta_{ji}\in\mathbb{R}$.

More in general, nonhomogeneous quasilinear systems can arise from scalar equations of arbitrary order. For example, consider the (2+1)-dimensional heat equation $u_t=u_{xx}+u_{yy}$. By performing coordinate transformations $u^1=u$, $u^2=u^1_x$, $u^3=u^1_y$ and interchanging variables $x$ with $t$ (or $y$ with $t$) \cite{tsarev1989hamilton}, we obtain the following system:
\begin{gather*}
u^1_t=u^2,\qquad
u^2_t=-u^1_x+u^3_y,\qquad
u^3_t=u^2_y.
\end{gather*}

\subsection{The cotangent covering}
The core content of this paper is to investigate the necessary conditions for the multidimensional quasilinear system \eqref{sys} to be the Hamiltonian system. Similarly to \cite{vergallo2023non}, we define these conditions as compatibility conditions and apply the theory of differential coverings \cite{krasil1989nonlocal} for computation. Since in \cite{vergallo2021homogeneous} the authors have already presented the procedure, we will provide a~brief overview of the theoretical framework in the multidimensional case.

Let us consider a multidimensional first-order quasilinear system
\begin{gather}\label{system19}
F^i=u^i_t-V^{i\alpha}_lu^l_\alpha-W^i=0.
\end{gather}
The vector function $\varphi=\varphi^i$ is said to be a \emph{symmetry} of the system \eqref{sys} if it satisfies the condition $l_F(\varphi)=0$, where $l_F$ is the Frechet derivative, expressed as
\begin{gather*}
(l_F)^i_j=\delta^i_jD_t-\big(V^{i\alpha}_{l,j}u^l_\alpha+W^i_{,j}\big)-V^{i\alpha}_jD_\alpha.
\end{gather*}
Moreover, the adjoint of $l_F$ denotes
\begin{gather*}
(l^*_F)^i_j=\big((l_F)^j_i\big)^*=-\delta^j_iD_t+\big(V^{j\alpha}_{i,l}u^l_\alpha-V^{j\alpha}_{l,i}u^l_\alpha-W^j_{,i}\big)+V^{j\alpha}_iD_\alpha.
\end{gather*}
The corresponding vector function $\psi=\psi_i=\frac{\delta b}{\delta u^i}$ such that $l^*_F(\psi)=0$ is called a \emph{cosymmetry} of the system \eqref{sys}.

We then introduce the so-called cotangent covering $T^*$ for the system~\eqref{system19}, which is the system $F=0$ and $l^*_F(p)=0$ where $p$ is an auxiliary odd dependent variable corresponding to a~cosymmetry. Explicitly, this is
\begin{gather}\label{cov}
u^i_t=V^{i\alpha}_lu^l_\alpha+W^i,\qquad
p_{i,t}=\big(V^{j\alpha}_{i,l}u^l_\alpha-V^{j\alpha}_{l,i}u^l_\alpha-W^j_{,i}\big)p_j+V^{j\alpha}_ip_{j,\alpha},
\end{gather}
where $p_{i,\alpha}=D_\alpha\psi_i$ and so on.

In \cite{kersten2004hamiltonian}, the following theorem is presented.
\begin{Theorem}
If the system \eqref{sys} admits the Hamiltonian formulation, then
\begin{gather}\label{eq:main-comp}
l_F\circ P=P^*\circ l_F^*,
\end{gather}
where the operator $P$ is a Hamiltonian structure for $F=0$.
\end{Theorem}
\begin{Remark}
\eqref{eq:main-comp} is necessary but not sufficient for \eqref{sys} to admit a Hamiltonian formulation.
\end{Remark}
The proof in \cite{kersten2004hamiltonian} is independent from the dimension of $X$, namely on the number of spatial variables, so it is not only true for one-dimensional systems but also in our case. Therefore, the following result immediately follows.
\begin{Proposition}
If \eqref{sys} is a Hamiltonian system with Hamiltonian structure \smash{$P^{ij}=P^{ij}_S\partial^S$}, then
\begin{gather}\label{eq:comp-V}
(l_F)^i_jP^{jl}(p_l)=0
\end{gather}
for $(u,p)$ the cotangent covering of $F$.
\end{Proposition}
\begin{proof}
If we apply the identity \eqref{eq:main-comp} to $p$, then the right-hand side vanishes and the left-hand side is~\eqref{eq:comp-V}.
\end{proof}

We call \eqref{eq:comp-V} the compatibility condition. It serves as a necessary condition (but not sufficient) for the multidimensional quasilinear system \eqref{sys} to be Hamiltonian. The condition \eqref{eq:comp-V} cannot guarantee that the operator~$P$ is Hamiltonian nor it does ensure that, even in the case when $P$ is a Hamiltonian operator, the system $F$ is Hamiltonian.

However, by computing the condition \eqref{eq:comp-V}, we can exclude certain operators. In other words, for a given quasilinear system, when searching for possible Hamiltonian structures, all the operators (Hamiltonian or not) that do not satisfy \eqref{eq:comp-V} should be discharged. Moreover, for a given operator satisfying the Hamiltonian conditions, it is also possible to search for compatible systems by computing condition \eqref{eq:comp-V}.

\subsection{Compatibility in the multidimensional case}
The next step is to compute the compatibility conditions. We assume that the operator $P$ is Hamiltonian and consider the Hamiltonian formulation of the first-order quasilinear system
\begin{gather*}
u^i_t=V^{i\alpha}_lu^l_\alpha+W^i=P^{ij}\frac{\delta H}{\delta u^j}.
\end{gather*}

Since we consider quasi-linear systems, we assume that $P$ is at most of differential order~1, and consequently assume that $H$ depends only on the variables $u$ (and not on higher-order jet variables). For a multidimensional nonhomogeneous quasilinear system of type 1+0, the corresponding Hamiltonian structure must be the form of $P=A+\omega$, where $\operatorname{deg}(A)=1$ and~$\operatorname{deg}(\omega)$=0. Similarly, the 1-order homogeneous quasilinear system corresponds to the homogeneous operator $P$ such that $\operatorname{deg}(P)$= 1, and the 0-order quasilinear system corresponds to the 0-degree operator.

We directly consider the compatibility conditions of (1+0)-order nonhomogeneous systems
\begin{equation}\label{nonsys}
u^i_t=V^{i\alpha}_lu^l_\alpha+W^i, \qquad i=1,2,\dots, n,
\end{equation}
which admit nonhomogeneous operators (satisfying the conditions \eqref{con1}, \eqref{con2} and \eqref{con3})
\begin{equation}\label{HS}
P^{ij}=g^{ij\alpha}\partial_\alpha+b^{ij\alpha}_ku^k_\alpha+\omega^{ij}
\end{equation}
as their Hamiltonian structures.
\begin{Theorem}\label{theorem8}
If a nonhomogeneous quasilinear system \eqref{nonsys} admits a Hamiltonian formulation with a Hamiltonian structure \eqref{HS}, the following conditions need to be satisfied:
\begin{subequations}
\begin{gather}
\omega^{ij}_{,s}W^s-\omega^{sj}W^i_{,s}-\omega^{is}W^j_{,s}=0,\label{syscon1}\\
\sum_{(\alpha,\beta)}\big(V^{j\alpha}_sg^{is\beta}-V^{i\alpha}_sg^{sj\beta}\big)=0,\label{syscon2}\\
g^{ij\beta}_{,s}V^{s\alpha}_l+g^{is\beta}\big(V^{j\alpha}_{s,l}\!-V^{j\alpha}_{l,s}\big)+g^{is\alpha}V^{j\beta}_{s,l}
+b^{is\alpha}_lV^{j\beta}_s-g^{sj\beta}V^{i\alpha}_{l,s}-g^{sj\beta}_{,l}V^{i\alpha}_s-b^{sj\alpha}_lV^{i\beta}_s
=0,\!\!\!\label{syscon3}\\
\sum_{(\alpha,\beta)}\big[g^{is\beta}\big(V^{j\alpha}_{s,l}-V^{j\alpha}_{l,s}\big)+b^{ij\beta}_sV^{s\alpha}_l-V^{i\alpha}_sb^{sj\beta}_l\big]=0,\label{syscon4}\\
g^{is\beta}\big(V^{j\alpha}_{s,lk}-V^{j\alpha}_{l,sk}\big)+g^{is\alpha}\big(V^{j\beta}_{s,kl}-V^{j\beta}_{k,sl}\big)+b^{ij\beta}_{k,h}V^{h\alpha}_l+b^{ij\alpha}_{l,h}V^{h\beta}_k +b^{ij\beta}_hV^{h\alpha}_{l,k}\notag \\
\qquad{}+b^{ij\alpha}_hV^{h\beta}_{k,l}+b^{is\beta}_k\big(V^{j\alpha}_{s,l}
-V^{j\alpha}_{l,s}\big)+b^{is\alpha}_l\big(V^{j\beta}_{s,k}-V^{j\beta}_{k,s}\big) -b^{hj\beta}_kV^{i\alpha}_{l,h}\notag \\
\qquad{}
-b^{hj\alpha}_lV^{i\beta}_{k,h}-b^{hj\beta}_{k,l}V^{i\alpha}_h-b^{hj\alpha}_{l,k}V^{i\beta}_h=0,\label{syscon5}\\
g^{ij\alpha}_{,s}W^s-g^{is\alpha}W^j_{,s}-W^i_{,s}g^{sj\alpha}+\omega^{is}V^{j\alpha}_s-V^{i\alpha}_s\omega^{sj}=0,\label{syscon6}\\
-g^{is\alpha}W^j_{,sl}+b^{ij\alpha}_{l,s}W^s+b^{ij\alpha}_sW^s_{,l}-b^{is\alpha}_lW^j_{,s}-W^i_{,s}b^{sj\alpha}_l \notag \\
\qquad{}+\omega^{ij}_{,s}V^{s\alpha}_l+\omega^{is}\big(V^{j\alpha}_{s,l}-V^{j\alpha}_{l,s}\big)-V^{i\alpha}_{l,s}\omega^{sj}-V^{i\alpha}_s\omega^{sj}_{,l}=0,\label{syscon7}
\end{gather}
\end{subequations}
where $\sum$ represents the cyclic summation over indices.
\end{Theorem}
\begin{proof}
We note that the cotangent covering of the system \eqref{nonsys} is \eqref{cov}. In order to obtain the necessary conditions, it is required to compute $l_F(P(p))=0$,
\begin{gather*}
l_F(P(p))=\big(\delta_{ij}D_t-V^{i\alpha}_{l,j}-W^i_{,j}-V^{i\alpha}_jD_\alpha\big)\big(g^{js\beta}p_{s,\beta}+b^{js\beta}_ku^k_\beta p_s+\omega^{js}p_s\big).
\end{gather*}
By setting the degree of $p_i$ to 0 and extracting the terms of degree 0, namely the term $p_j$, we can obtain \eqref{syscon1}. Indeed, this is also the compatibility condition of the 0-order system~${u^i_t=W^i(u)}$ with respect to the 0-order ultralocal Hamiltonian operator. Continuing with the extraction of terms of degree 2, the terms $p_{j,\alpha\beta}$, $u^l_\alpha p_{j,\beta}$, $u^l_{\alpha\beta} p_j$, and $u^l_\alpha u^k_\beta p_j$ correspond to \eqref{syscon2}--\eqref{syscon4}, and~\eqref{syscon5}, respectively. Similarly, these conditions are the compatibility conditions for first-order homogeneous systems admitting hydrodynamic-type operators as their Hamiltonian structures. Finally, extracting the remaining terms $p_{j,\alpha}$, $u^l_\alpha p_j$, we can obtain \eqref{syscon6}--\eqref{syscon7}.
\end{proof}

\begin{Remark}
Assuming that $\omega$ is non-degenerate, then \eqref{syscon1} is equivalent to $\tilde{\nabla}^iW^j=\tilde{\nabla}^jW^i$ (this is both the notation and the result of \cite{vergallo2023non}). \eqref{syscon4} contains three sets of conditions (one for $\alpha=\beta=x$, one for $\alpha=\beta=y$, and one for $\alpha=x$, $\beta=y$). From this observation, we can reformulate \eqref{syscon4} as three ``one-dimensional'' conditions
\begin{gather*}
g^{is(\kappa)}\big(V^{j(\kappa)}_{s,l}-V^{j(\kappa)}_{l,s}\big)+b^{ij(\kappa)}_sV^{s(\kappa)}_l-V^{i(\kappa)}_sb^{sj(\kappa)}_l=0,\qquad(\kappa)=1,2,3,
\end{gather*}
for \big(\smash{$g^{(1)}=g^x$}, \smash{$b^{(1)}=b^x$}, \smash{$V^{(1)}=V^x$}\big), \big(\smash{$g^{(2)}=g^y$}, \smash{$b^{(2)}=b^y$}, \smash{$V^{(2)}=V^y$}\big) and \big($g^{(3)}=g^x+g^y$, \smash{$b^{(3)}=b^x+b^y$}, \smash{$V^{(3)}=V^x+V^y$}\big). Therefore, if $g^x$, $g^y$ \emph{and} $g^x+g^y$ are all non-degenerate, \eqref{syscon4} can be written as \smash{$\nabla^{(\kappa)}_i V^{j(\kappa)}_k=\nabla^{(\kappa)}_k V^{j(\kappa)}_i$}.
\end{Remark}

\subsection{Compatibility and Poisson cohomology}\label{sec:Pois}
It is possible to offer a different interpretation of the property \eqref{eq:main-comp}, which clearly highlights the relation between sufficiency and necessity of the compatibility condition~\eqref{eq:comp-V}. Note that in this paragraph we denote both the Poisson bivector and the Hamiltonian operator defining it with the same letter~$P$, to keep a consistent notation with the previous formulae. The notion of Poisson cohomology was first introduced by Lichnerowicz \cite{lichnerowicz1977varietes} in the finite-dimensional setting; it has become a fundamental tool in the study of infinite-dimensional Hamiltonian systems, where the Poisson bivector and the Hamiltonian operator can be identified (see, for instance, \cite{dubrovin2001normal,getzler}). Let us briefly recall the notion of Poisson cohomology for a Hamiltonian operator.

Given a Poisson bivector $P\in\Lambda^2$, the condition $[P,P]=0$ guarantees, by the properties of the Schouten bracket, that its adjoint action on poly-vector fields ${\rm d}_P=[P,\cdot]\colon\Lambda^k\to\Lambda^{k+1}$ squares to 0, turning the space of poly-vector fields into a cochain complex known as the \emph{Poisson--Lichnerowicz complex}. The linear operator ${\rm d}_P$ is therefore called the \emph{Poisson differential},
\begin{gather*}
0\to\Lambda^0\xrightarrow{{\rm d}_P}\Lambda^1\xrightarrow{{\rm d}_P}\Lambda^2\xrightarrow{{\rm d}_P}\cdots\xrightarrow{{\rm d}_P}\Lambda^p\xrightarrow{{\rm d}_P}\cdots.
\end{gather*}
Note that, in the infinite-dimensional case, the Poisson--Lichnerowicz complex does not terminate, while $0$-vector fields are identified with local functionals. The cohomology of such complex is referred to as the \emph{Poisson cohomology}:
\begin{gather*}
H(P)=\bigoplus\limits_{p\geq 0}H^p(P)=\bigoplus\limits_{p\geq 0}\frac{\operatorname{Ker}{\rm d}_P\colon\Lambda^p\to\Lambda^{p+1}}{\operatorname{Im} {\rm d}_P\colon\Lambda^{p-1}\to\Lambda^p}=\bigoplus\limits_{p\geq 0}\frac{Z^p}{B^p}.
\end{gather*}
As in any cochain complex, the elements of $Z^p$ are called cocycles, while the elements of $B^p$ are called coboundaries. We also introduce a second grading on $H^p$ (resp.\ $Z^p$, $B^p$, $\Lambda^p$), corresponding to the differential order of the elements, by the rule $\deg u=0$, $\deg D_\alpha=1$, $\deg u^i_Q=q_1+\cdots+q_D$. The corresponding homogeneous components of $H^p$ (resp.\ for $Z,B,\ldots$) are denoted by $H^p_d$, $\big(Z^p_d,\ldots\big)$, $d\geq 0$.

A vector field that preserves the Hamiltonian structure (in the sense $\{X(F),G\}+\{F,X(G)\}=X(\{F,G\})$, where $\{\cdot,\cdot\}$ is the Poisson bracket defined in~\eqref{eq:pb-def}) is called a \emph{symmetry} of the structure and it is a $1$-cocycle. Indeed, the condition can equivalently be written as $\mathcal{L}_X(P)=0$ and the Lie derivative of~$P$ along the vector field~$X$ is, by the definition of Schouten bracket,~$-{\rm d}_PX$. On the other hand, a Hamiltonian vector field is obtained by the adjoint action of~$P$ on a local functional (the \emph{Hamiltonian functional}), $X_H=P(\delta H)=-[P,H]$, and it is therefore a $1$-coboundary. The elements of the first cohomology group $H^1(P)$ are the so-called \emph{non-Hamiltonian symmetries}. They are vector fields that preserve the Hamiltonian structure without being Hamiltonian.

\begin{Proposition}
Let $X$ be the evolutionary vector field corresponding to the system $F^i=u^i_t-X^i=0$, where $\{X^i\}$ is the characteristics of $X$, and $P$ the Poisson bivector defined by a local Hamiltonian operator. Then condition \eqref{eq:main-comp} is equivalent to $X$ being a $1$-cocycle in the Poisson--Lichnerowicz complex of $P$ $($namely, $X$ is a symmetry of $P)$.
\end{Proposition}
\begin{proof}
This observation is essentially computational. First, let us explicitly write condition~\eqref{eq:main-comp} in terms of $F^i=u^i_t-X^i$ and $P=P^{ij}_S\partial^S$. For any covector $\psi=\{\psi_i\}$, we have
\begin{align}
(l_F(P\psi))^j&=D_t\big(P^{ji}_SD^S\psi_i\big)-\frac{\partial X^j}{\partial u^m_Q}D^Q\big(P^{mi}_SD^S\psi_i\big)\nonumber\\
&=\frac{\partial P^{ji}_S}{\partial u^m_Q}\big(D^Q X^m\big)D^S\psi_i+P^{ji}_SD_tD^S\psi_i-\frac{\partial X^j}{\partial u^m_Q}D^Q\big(P^{mi}_SD^S\psi_i\big),\label{eq:proof-S1}
\end{align}
where we use the multi-index notation $D^S$ to denote $D_1^{s_1}D_2^{s_2}\cdots D^{s_D}_{D}$.
Similarly,
\begin{align}
(P^*(l_F^*\psi))^j&=(-D)^S\left(P^{ij}_S\left(-D_t\psi_i-(-D)^Q\frac{\partial X^m}{\partial u^i_Q}\psi_m\right)\right)\nonumber\\
&=-\big[(-D)^SP^{ij}D_t\psi_i\big]-(-D)^S\left[ P^{mj}_S(-D)^Q\left(\frac{\partial X^i}{\partial u^m_Q}\psi_i\right)\right].\label{eq:proof-S2}
\end{align}
Now we subtract \eqref{eq:proof-S1} from \eqref{eq:proof-S2}, utilize the commutativity between spatial and time derivatives $[D_\alpha,D_t]=0$ and the skewsymmetry property for $P$, namely $P^{ji}_SD^S=-\bigl(-D^S\bigr)P^{ij}_S$, to obtain that condition \eqref{eq:main-comp} is, explicitly,
\begin{equation}\label{eq:proof-comp}
\left[P^{jm}_SD^S(-D)^Q\frac{\partial X^i}{\partial u^m_Q}+\frac{\partial X^j}{\partial u^m_Q}D^QP^{mi}_SD^S-\big(D^QX^m\big)\frac{\partial P^{ji}_S}{\partial u^m_Q}D^S\right]\psi_i=0
\end{equation}
for any $\psi=\{\psi_i\}$.

On the other hand, the cocycle condition ${\rm d}_PX=0$ can be expressed using the formalism of PVAs \cite{casati2015deformations} as
\[
\big\{X(u^i)_\lambda u^j\big\}-\big\{u^i{}_\lambda X(u^j)\big\}-X\big(\big\{u^i{}_\lambda u^j\big\}\big)=0
\]
for all the pairs $(i,j)$ and a $\lambda$ bracket $\{u^i{}_{\lambda}u^j\}=P^{ji}_S\lambda^S$. This can be explicitly computed with the master formula \eqref{eq:master}, obtaining
\begin{gather}\label{eq:proof-symPVA}
P^{jm}_S(\lambda+\partial)^S(-\lambda-\partial)^Q\frac{\partial X^i}{\partial u^m_Q}+\frac{\partial X^j}{\partial u^m_Q}(\lambda+\partial)^Q P^{mi}_S\lambda^S-\big(D^Q X^m\big)\frac{\partial P^{ji}_S}{\partial u^m_Q}\lambda^S=0.
\end{gather}
Now it is enough to observe that the total derivatives $D$ in \eqref{eq:proof-comp} may act on both $\psi_i$ and differential expressions with $u$ variables. Denoting a total derivative $D^S$ acting on $\psi_i$ by $\lambda^S$ and a total derivative acting on a differential expression of the $u$ variables by $\partial^S$ (this second case, indeed, coincides with the definition of $\partial$ we used in presenting multidimensional PVAs), by the Leibniz rule we have $D^S=(\partial+\lambda)^S$. This replacement gives exactly \eqref{eq:proof-symPVA}.
\end{proof}

It is now clear that the sufficiency of condition \eqref{eq:main-comp} (and of its equivalent form \eqref{eq:comp-V}) is closely related to the triviality of $H^1(P)$. Indeed, the existence of any nontrivial cohomology class denotes the existence of compatible evolutionary systems without Hamiltonian formulation.

The study of the Poisson cohomology for multidimensional Hamiltonian operators has shown that, in general, such an object is highly non-trivial \cite{ccs17,ccs18,casati2015deformations}; we can therefore expect that we will need to discard some solutions of system \eqref{eq:comp-V} to identify \emph{bona fide} Hamiltonian systems. An independent knowledge of the Poisson cohomology of the operator can help us in this task. Conversely, finding the solutions of \eqref{eq:comp-V} that fail to be Hamiltonian provides an explicit method to compute the (first) Poisson cohomology of $P$.

While the Poisson cohomology for homogeneous operators, and in particular for structures of hydrodynamic type, has been previously studied in the aforementioned references, the one for nonhomogeneous ones has not been computed yet, at best of our knowledge. We will hereinafter demonstrate how $H^1$ for a nonhomogeneous operator of order $1{+}0$ is closely related to the bi-Hamiltonian cohomology of the two homogeneous components ${\rm BH}(P_0,P_1)$ \cite{lz05}. Finally, in the following section we will show that in the $N=D=2$ case (extensively discussed in \cite{casati2015deformations}) it is isomorphic to the cohomology of the leading order operator.
\begin{Proposition}\label{prop:biham}
Let us consider a $1{+}0$-order Poisson bivector $P=P_1+P_0$ with nondegenerate~$P_0$ $($the underscripts, here and in the following, denote the \emph{differential degree} of the objects$)$. Then~${H^1_{\leq 1}(P)={\rm BH}^1_{1}(P_1,P_0)}$.
\end{Proposition}
\begin{proof}
Let ${\rm d}_P$ (resp.\ ${\rm d}_0$ and ${\rm d}_1$) the Poisson differentials associated to the bivectors $P$, $P_0$ and $P_1$, and note that $\deg P_0=0$, $\deg P_1=1$. The first observation is that ${\rm d}_0$ and ${\rm d}_1$ form a \emph{differential bi-complex}, namely ${\rm d}_0^2=0$, ${\rm d}_1^2=0$, ${\rm d}_0{\rm d}_1=-{\rm d}_1{\rm d}_0$. The first two identities are obvious because $P_0$ and $P_1$ are Poisson bivectors, while the third one corresponds to their compatibility in the bi-Hamiltonian sense. All the three identities can be obtained by considering the homogeneous parts of ${\rm d}_P^2=({\rm d}_0+{\rm d}_1)^2=0$. Now let us consider elements of $Z^1_{\leq 1}(P)$. They must be vector fields $X=X_0+X_1$ such that ${\rm d}_P X=0$; expanding the condition by differential degree, we have
\begin{gather}\label{eq:proof-BH1}
 {\rm d}_0 X_0=0,\qquad {\rm d}_0X_1+{\rm d}_1X_0=0,\qquad {\rm d}_1X_1=0.
\end{gather}
In particular, then, $X_0\in Z^1_0(P_0)$ and $X_1\in Z^1_1(P_1)$. It has been known for long time \cite{getzler} that, for invertible bivectors $P_0$ (and trivial de~Rham cohomology of the target space $M$, which we always assume is a ball) $H^1(P_0)=0$, so there exists a local functional $h_0$ such that $X_0={\rm d}_0h_0$; on the other hand, $H^1_1(P_1)$ is in general non trivial: therefore, $X_1={\rm d}_1 k_0 + \xi_1$, where $k_0$ is a local functional (that, for homogeneity reasons, must be of differential order~0) and $\xi_1\in H^1_1(P_1)$. We can now replace the expressions for $X_0$ and $X_1$ in the second equation of \eqref{eq:proof-BH1}. This gives us the identity
\begin{gather*}
{\rm d}_0{\rm d}_1k_0+{\rm d}_0\xi_1+{\rm d}_1{\rm d}_0h_0={\rm d}_0=({\rm d}_1(k_0-h_0)+\xi_1)=0,
\end{gather*}
which, thanks to the triviality of $H^1(P_0)$, is
\begin{equation}\label{eq:proof-BH2}
{\rm d}_1(k_0-h_0)+\xi_1={\rm d}_0 h_1.
\end{equation}
This means that, for $X=X_0+X_1$, we can write
\begin{gather*}
X=({\rm d}_0+{\rm d}_1)h_0+{\rm d}_1(k_0-h_0)+\xi_1={\rm d}_P h_0 + {\rm d}_0h_1.
\end{gather*}
However, this expression shows that the elements of \smash{$H^1_{\leq1}(P)$} are of the form ${\rm d}_0h_1={\rm d}_1 g_0+\xi_1$, where we simply denote $k_0-h_0=g_0$. In particular, we have ${\rm d}_0^2h_1=0$ and, thanks to \eqref{eq:proof-BH2}, ${\rm d}_1{\rm d}_0h_1=0$, so $H^1_{\leq 1}(P)= (\operatorname{Ker} {\rm d}_0\cap\operatorname{Ker} {\rm d}_1 )^1_{1}$. This is exactly the definition of ${\rm BH}^1_{1}(P_0,P_1)$, see, for instance, \cite{dubrovin2001normal,lz05}.
\end{proof}

\section{Examples}
\subsection[N=D=2 systems]{$\boldsymbol{N=D=2}$ systems}
\subsubsection{The nonhomogeneous Hamiltonian structures}
In this section, we focus on the study of the Poisson brackets of 1+0 hydrodynamic type in the case $D=N=2$, along with their compatible systems. In this case we consider the three possible normal forms for the Hamiltonian structures of hydrodynamic type \cite{ferapontov2011classification,mokhov2008classification}
\begin{gather*}
P_1=\begin{pmatrix} 1&0\\0&0\\ \end{pmatrix} \frac{\rm d}{{\rm d}x}+\begin{pmatrix} 0&0\\0&1\\ \end{pmatrix} \frac{\rm d}{{\rm d}y},\qquad
P_2=\begin{pmatrix} 0&1\\1&0\\ \end{pmatrix} \frac{\rm d}{{\rm d}x}+\begin{pmatrix} 0&0\\0&1\\ \end{pmatrix} \frac{\rm d}{{\rm d}y},\\
P_3=\begin{pmatrix} 2u&v\\v&0\\ \end{pmatrix} \frac{\rm d}{{\rm d}x}+\begin{pmatrix} 0&u\\u&2v\\ \end{pmatrix} \frac{\rm d}{{\rm d}y}+\begin{pmatrix} u_x&u_y\\v_x&v_y\\ \end{pmatrix}
\end{gather*}
and the ultralocal structure
\smash{$
\omega=\bigl(\begin{smallmatrix} 0&f(u,v)\\-f(u,v)&0\\ \end{smallmatrix}\bigr)$},
which is Hamiltonian for any choice of $f$.

Now one should investigate which kind of nonhomogeneous structures are allowed starting with these pairs of Hamiltonian structures $(P_i,\omega)$. The research method involves applying Theorem \ref{theor5} to obtain compatibility conditions, leading to the following results.
\begin{Proposition}
The operators $P=P_1+\omega$ and $P_2+\omega$ are Hamiltonian if and only if $f=\eta$ $(\eta$ is a constant$)$. The operator $P_3+\omega$ is Hamiltonian if and only if $f=0$.
\end{Proposition}

\subsubsection{Compatibility conditions and Hamiltonian equations}\label{sec:res}
We can build upon the results from the previous subsection and search for compatible quasilinear systems $u^i_t=V^{i\alpha}_lu^l_\alpha+W^i$. The corresponding method involves computing the seven conditions outlined in Theorem~\ref{theorem8}.

\begin{Proposition}\label{proposition4.2}
The solutions of the compatibility system corresponding to the Hamiltonian operator
\begin{gather*}
P_1+\omega=\begin{pmatrix} 1&0\\0&0\\ \end{pmatrix} \frac{\rm d}{{\rm d}x}+\begin{pmatrix} 0&0\\0&1\\ \end{pmatrix} \frac{\rm d}{{\rm d}y}+\begin{pmatrix} 0&\eta\\-\eta&0\\ \end{pmatrix}
\end{gather*}
are as follows:
\begin{gather*}
V^{2x}_1=V^{1y}_2=0, \qquad V^{2x}_2=\alpha, \qquad V^{1y}_1=\beta,\qquad
V^{1x}_2=V^{2y}_1=\frac{\partial^2M}{\partial u\partial v}, \\ V^{1x}_1=\frac{\partial^2M}{\partial u^2}+\alpha, \qquad V^{2y}_2=\frac{\partial^2M}{\partial v^2}+\beta,\qquad
W^1=\eta\frac{\partial M}{\partial v}, \qquad W^2=-\eta\frac{\partial M}{\partial u},
\end{gather*}
where $M$ is an arbitrary function of $u$ and $v$, and $\alpha$, $\beta$ are arbitrary constants. Therefore, any candidate Hamiltonian system with Hamiltonian $P_1+\omega$ must be of the form
\begin{gather}
u_t=\left(\frac{\partial^2 M}{\partial u^2}+\alpha\right)u_x+\beta u_y+\frac{\partial^2 M}{\partial u\partial v}v_x+\eta\frac{\partial M}{\partial v},\nonumber\\
v_t=\frac{\partial^2 M}{\partial u\partial v}u_y+\alpha v_x+\left(\frac{\partial^2 M}{\partial v^2}+\beta\right)v_y-\eta\frac{\partial M}{\partial u}.\label{eq:solcomp1}
\end{gather}
\end{Proposition}
\begin{Proposition}\label{proposition4.3}
The solutions of the compatibility system corresponding to the Hamiltonian operator
\begin{gather*}
P_2+\omega=\begin{pmatrix} 0&1\\1&0\\ \end{pmatrix} \frac{\rm d}{{\rm d}x}+ \begin{pmatrix} 0&0\\0&1\\ \end{pmatrix} \frac{\rm d}{{\rm d}y}+\begin{pmatrix} 0&\eta\\-\eta&0\\ \end{pmatrix}
\end{gather*}
are as follows:
\begin{gather*}
 V^{1y}_2=0, \qquad V^{1y}_1=\alpha ,
 \qquad V^{2x}_1=\frac{\partial^2M}{\partial u^2} ,\qquad
 V^{2y}_2=\frac{\partial^2M}{\partial v^2}+\alpha, \qquad
 V^{2x}_2=\frac{\partial^2M}{\partial u \partial v}+\beta,\\ V^{1x}_1=\frac{\partial^2M}{\partial u \partial v}+\beta,\qquad
 V^{2y}_1=\frac{\partial^2M}{\partial u \partial v} \qquad, V^{1x}_2=\frac{\partial^2M}{\partial v^2}, \qquad
 W^1=\eta\frac{\partial M}{\partial v} ,
\\ W^2=-\eta\frac{\partial M}{\partial u} ,
\end{gather*}
where $M$ is an arbitrary function of $u$ and $v$, and $\alpha$, $\beta$ are arbitrary constants. Therefore, any candidate Hamiltonian system with Hamiltonian $P_2+\omega$ must be of the form
\begin{gather}
u_t=\left(\frac{\partial^2 M}{\partial u\partial v}+\beta\right)u_x+\alpha u_y+\frac{\partial^2 M}{\partial v^2}v_x+\eta\frac{\partial M}{\partial v},\nonumber\\
v_t=\frac{\partial^2 M}{\partial u^2}u_x+\frac{\partial^2 M}{\partial u\partial v}u_y+\left(\frac{\partial^2 M}{\partial u\partial v}+\beta\right)v_x+\left(\frac{\partial^2 M}{\partial v^2}+\alpha\right)v_y-\eta\frac{\partial M}{\partial u}.\label{eq:solcomp2}
\end{gather}
\end{Proposition}
\begin{Proposition}\label{proposition4.4}
The solutions of the compatibility system corresponding to the Hamiltonian operator
\begin{gather*}
P_3=\begin{pmatrix} 2u&v\\v&0\\ \end{pmatrix} \frac{\rm d}{{\rm d}x}+\begin{pmatrix} 0&u\\u&2v\\ \end{pmatrix} \frac{\rm d}{{\rm d}y}+\begin{pmatrix} u_x&u_y\\v_x&v_y\\ \end{pmatrix}
\end{gather*}
are as follows:
\begin{gather*}
 V^{1x}_1=2u\frac{\partial^2M}{\partial u^2}+\frac{\partial M}{\partial u}+v\frac{\partial^2M}{\partial u \partial v}, \qquad V^{1x}_2=2u\frac{\partial^2M}{\partial u \partial v}+v\frac{\partial^2M}{\partial u \partial v},\qquad
 V^{2x}_1=v\frac{\partial^2M}{\partial u^2}, \\ V^{2x}_2=\frac{\partial M}{\partial u}+v\frac{\partial^2M}{\partial u \partial v}, \qquad
 V^{1y}_1=\frac{\partial M}{\partial v}+u\frac{\partial^2M}{\partial u \partial v}, \qquad V^{1y}_2=u\frac{\partial^2M}{\partial v^2}, \\
 V^{2y}_1=u\frac{\partial^2M}{\partial u^2}+2v\frac{\partial^2M}{\partial u \partial v} ,
 \qquad V^{2y}_2=\frac{\partial M}{\partial v}+u\frac{\partial^2M}{\partial u\partial v}+2v\frac{\partial^2M}{\partial v^2},
\end{gather*}
where $M$ is an arbitrary function of $u$ and $v$. Any system compatible with $P_3+\omega$ is therefore of the form
\begin{gather*}
u_t=\left(2u\frac{\partial^2 M}{\partial u^2}+v\frac{\partial^2 M}{\partial u\partial v}+\frac{\partial M}{\partial u}\right)u_x+\left(u\frac{\partial^2 M}{\partial u\partial v}+\frac{\partial M}{\partial v}\right) u_y+\left(2u\frac{\partial^2 M}{\partial u\partial v}+v\frac{\partial^2M}{\partial u\partial v}\right)v_x\\
\phantom{u_t=}{}+
u\frac{\partial^2 M}{\partial v^2}v_y,\\
v_t=v\frac{\partial^2 M}{\partial u^2}u_x+\left(u\frac{\partial^2 M}{\partial u^2}+2v\frac{\partial^2 M}{\partial u\partial v}\right) u_y+\left(v\frac{\partial^2 M}{\partial u\partial v}+\frac{\partial M}{\partial u}\right)v_x\\
\phantom{v_t=}{}
+\left(u\frac{\partial^2 M}{\partial u\partial v}+2v\frac{\partial^2 M}{\partial v^2}+\frac{\partial M}{\partial v}\right)v_y.
\end{gather*}
\end{Proposition}

To demonstrate our results, let us briefly present as an example the proof of Proposition~\ref{proposition4.2}.
\begin{proof}[Proof of Proposition~\ref{proposition4.2}]\allowdisplaybreaks
By computing the conditions in Theorem~\ref{theorem8}, we can obtain the following results:
\begin{subequations}
\begin{gather}
 W^1_{,1}+W^2_{,2}=0,\label{re1}\\
 V^{2x}_1=V^{1y}_2=0,\label{re2}\\
 V^{2x}_{2,1}=V^{2x}_{2,2}=V^{1y}_{1,1}=V^{1y}_{1,2}=0,\label{re3}\\
 V^{1x}_2=V^{2y}_1,\label{re4}\\
 V^{1x}_{1,2}=V^{1x}_{2,1},\label{re5}\\
 V^{2y}_{1,2}=V^{2y}_{2,1},\label{re6}\\
 W^1_{,1}=\eta V^{1x}_2,\label{re7}\\
 W^2_{,2}+\eta V^{2y}_1=0,\label{re8}\\
 W^1_{,2}+\eta V^{1y}_1=\eta V^{2y}_2,\label{re9}\\
 W^2_{,1}+\eta V^{1x}_1=\eta V^{2x}_2.\label{re10}
\end{gather}
\end{subequations}
From \eqref{re1}, it can be deduced that
$W^1=\frac{\partial P}{\partial v}$,
$W^2=-\frac{\partial P}{\partial u}$,
where $P=P(u,v)$ is an arbitrary function.
Upon observing \eqref{re3}, we deduce that
$V^{2x}_2=\alpha$,
$V^{1y}_1=\beta$,
where $\alpha$ and $\beta$ are arbitrary constants.
Continuing with \eqref{re4}--\eqref{re6}, we can obtain
$V^{1x}_1=\frac{\partial^2 R}{\partial u^2}$,
$V^{2y}_2=\frac{\partial^2 R}{\partial v^2}$,
$V^{1x}_2=V^{2y}_1=\frac{\partial^2 R}{\partial u\partial v}$,
where $R=R(u,v)$ is an arbitrary function.

Subsequently, substituting these obtained results into the remaining equations \eqref{re7}--\eqref{re10}, we can establish the following relationships:
\begin{gather*}
\frac{\partial^2 R}{\partial u\partial v}=\frac{1}{\eta}\frac{\partial^2P}{\partial u\partial v},\qquad
\frac{\partial^2 R}{\partial u^2}=\frac{1}{\eta}\frac{\partial^2P}{\partial u^2}+\alpha,\qquad
\frac{\partial^2 R}{\partial v^2}=\frac{1}{\eta}\frac{\partial^2P}{\partial v^2}+\beta.
\end{gather*}
Finally, by rescaling $P$ to write $M=\frac{1}{\eta}P$, we obtain \eqref{eq:solcomp1}.
\end{proof}

In Propositions \ref{proposition4.2}--\ref{proposition4.4}, we have obtained the \emph{candidate} Hamiltonian equations with Hamiltonian structure respectively $P_1+\omega$, $P_2+\omega$, and $P_3+\omega$. However, we know that in principle not all the solutions we found are \emph{bona fide} Hamiltonian equations. Indeed, we can directly compare~\eqref{eq:solcomp1} with the corresponding Hamiltonian equation for a Hamiltonian $h(u,v)$, which is of the form
\begin{gather*}
u_t=\frac{\partial^2 h}{\partial u^2}u_x+\frac{\partial^2 h}{\partial u\partial v}v_x+\eta\frac{\partial h}{\partial v},\qquad
v_t=\frac{\partial^2 h}{\partial u\partial v}u_y+\frac{\partial^2 h}{\partial v^2}v_y-\eta\frac{\partial h}{\partial u}
\end{gather*}
and find that the terms containing the two constant $\alpha$ and $\beta$ cannot be obtained as Hamiltonian equations.

Similarly, upon examining the remaining two operators, it can be observed that for both~${P_1\!+\!\omega}$ and $P_2+\omega$ exist non-Hamiltonian solutions, while all the solutions of the compatibility condition for $P_3$ are Hamiltonian.

\subsubsection{Agreement with the Poisson cohomology}
Let us recall from Section \ref{sec:Pois} that the non-Hamiltonian systems selected by the compatibility conditions are in one-to-one correspondence to elements of ${\rm BH}^1_1(P_i+\omega)$, and such elements are of the form ${\rm d}_i g_0+\xi_1$ (${\rm d}_i:={\rm d}_{P_i}$), with $\xi_1\in H^1_1(P_i)$. The relevant results about the first cohomology group for the Poisson bracket defined by $P_1$, $P_2$, $P_3$ have been presented in~\cite{casati2015deformations}. We have
\begin{gather*}
H^1_0(P_1)\cong \mathbb{R}^2, \qquad H^1_0(P_2)\cong \mathbb{R}^2, \qquad H^1_0(P_3)\cong 0, \\
H^1_1(P_1)\cong \mathbb{R}^2, \qquad H^1_1(P_2)\cong \mathbb{R}^2 , \qquad H^1_1(P_3)\cong 0.
\end{gather*}
For both $P_1$ and $P_2$, we have that a compatible (but not Hamiltonian) vector field must be of the form
$X={\rm d}_1 g_0+\xi_1$
for $\xi^1\in H^1_1$ and a function $g_0=g_0(u,v)$, subject to the condition~${{\rm d}_0\big({\rm d}_1g_0+\xi^1\big)=0}$. The explicit form for a representative in the cohomology classes $H^1_1(P_1)$ and $H^1_1(P_2)$ are, respectively,
\begin{gather*}
\xi_1^{(1)}=\begin{pmatrix} \beta u_y\\ \alpha v_x\end{pmatrix},\qquad\xi_1^{(2)}=\begin{pmatrix} \alpha (u_y-v_x)\\-\beta u_y\end{pmatrix}
\end{gather*}
for $(\alpha,\beta)\in\mathbb{R}^2$. This choice of names for the constants matches the results of Section~\ref{sec:res}. Explicitly solving \smash{${\rm d}_0\big({\rm d}_1 g_0^{(i)}+\xi^{(i)}_1\big)=0$}, we find $g_0$ and
\begin{gather*}
X^{(1)}=\begin{pmatrix}\alpha u_x+\beta u_y\\\alpha v_x+\beta v_y\end{pmatrix}, \qquad X^{(2)}=\begin{pmatrix}\alpha u_y+\beta u_x\\ \alpha v_y+\beta v_x\end{pmatrix},
\end{gather*}
a 2-dimensional space of solutions exactly corresponding to the non-Hamiltonian solutions in~\eqref{eq:solcomp1} and \eqref{eq:solcomp2}.

It should be noticed that the non-Hamiltonian systems $u^i_t=X^{(1)i}$, corresponding to nonzero~$\alpha$ and $\beta$ constants in \eqref{eq:solcomp1} and \eqref{eq:solcomp2}, can be made trivial by a linear change of the independent variables $(x,y,t)$, respectively ($t\mapsto t$, $x\mapsto x+\alpha t$, $y\mapsto y+\beta t$) and $(t\mapsto t, x\mapsto x+\beta t, y\mapsto y+\alpha t)$. However, this change of coordinates involves also the time variable, so it is outside the scope of geometric description captured by the theory of Poisson cohomology, describing evolutionary systems as time evolution of functions on a manifold and hence involving only spatial variables and dependent variables (see, for example,~\cite{ccs17}).

On the other hand, $H^1_1(P_3)$ is trivial, so the only possible non-Hamiltonian vector field must be of the form $X={\rm d}_1g_0$ and such that ${\rm d}_0{\rm d}_1g_0=0$. This also means that ${\rm d}_0g_0\in H^1_0(P_3)$, which is trivial too, implying ${\rm d}_0g_0\in B^1_0(P_1)$. However, ${\rm d}_1$ is an operator of differential order $1$, therefore $B^1_0(P_1)=0$ and $g_0\in\ker {\rm d}_0$, i.e., $g_0=\mathrm{const}$. Finally, constants are in the kernel of ${\rm d}_1$, too, which means that non-Hamiltonian solutions of the compatibility equation do not exist.

\subsection{3-waves system}
Let us consider a real reduction of the 3-waves system ($N=3$ in \eqref{N-W}), i.e., $u_{12}=-u_{21}=u^3$, etc., as follows
\begin{gather}
u^1_t=au^1_x+du^1_y+(b-c)u^2u^3,\nonumber\\
u^2_t=bu^2_x+eu^2_y+(c-a)u^3u^1,\nonumber\\
u^3_t=cu^3_x+fu^3_y+(a-b)u^1u^2,\label{3-waves}
\end{gather}
where $a$, $b$, $c$, $d$, $e$, $f$ are parameters and distinct from each other.

We need to search for the Hamiltonian structure and the Hamiltonian function of the system~\eqref{3-waves}. According to the conditions in Theorem~\ref{theorem8}, the compatibility solutions of the Hamiltonian operator \eqref{HS} can be deduced.
\begin{Proposition}\label{prop:3W}
The system \eqref{3-waves} is compatible with a Hamiltonian operator $P$ if and only~if%
\begin{align} \label{eq:cond3W}
d=a p + q,\qquad e=b p+q,\qquad f=cp+q
\end{align}
for arbitrary $p, q \in \mathbb{R}$ and
\begin{equation}\label{3-W-S}
P=\begin{pmatrix} S&0&0\\0&S&0\\0&0&S\\ \end{pmatrix} \partial_x+\begin{pmatrix} pS&0&0\\0&pS&0\\0&0&pS\\ \end{pmatrix} \partial_y+\begin{pmatrix} 0&Su^3&-Su^2\\-Su^3&0&Su^1\\Su^2&-Su^1&0\\ \end{pmatrix},
\end{equation}
where $S$ is another arbitrary constant.
\end{Proposition}
\begin{proof}
By performing calculations \eqref{syscon2}--\eqref{syscon5}, we can deduce that $g^{ij\alpha}=0$ ($i\neq j$), and $g^{ii\alpha}$ depends only on $u^i$. Additionally, we obtain $b^{ij\alpha}_l = 0$ ($i\neq j$ or $i\neq l$), and~$b^{ii\alpha}_i$ depends only on~$u^i$. Continuing with the calculation \eqref{syscon6}, we can further deduce that
\begin{gather*}
g^{iix}=C_i,\qquad g^{iiy}=D_i,\qquad
\omega^{12}=\frac{C_1(c-a)+C_2(b-c)}{b-a}u^3,\\
\omega^{13}=\frac{C_1(a-b)+C_3(b-c)}{c-a}u^2,\qquad
\omega^{23}=\frac{C_2(a-b)+C_3(c-a)}{c-b}u^1,
\end{gather*}
where $C_i$ and $D_i$ are arbitrary constants.
Finally, utilizing equations \eqref{syscon1} and \eqref{syscon7}, it can be deduced that
\begin{gather*}
b^{ii\alpha}_i=0,\qquad C_1=C_2=C_3=S,\qquad
D_1=D_2=D_3=S\frac{f-e}{c-b}=S\frac{f-d}{c-a}=S\frac{d-e}{a-b},
\end{gather*}
where $S$ is an arbitrary constant.

Therefore, a constraint must be imposed, namely
$
\frac{f-e}{c-b}=\frac{f-d}{c-a}=p$.
Solutions of this constraint depend on an additional arbitrary parameter $q$ as in \eqref{eq:cond3W}.
By synthesizing the aforementioned results, we obtain the operator \eqref{3-W-S}, which satisfies the Hamiltonian conditions.
\end{proof}

According to Proposition \ref{prop:3W}, the compatible two-dimensional 3-waves system has the form
\begin{gather}
u^1_t=au^1_x+(pa+q)u^1_y+(b-c)u^2u^3,\nonumber\\
u^2_t=bu^2_x+(pb+q)u^2_y+(c-a)u^3u^1,\nonumber\\
u^3_t=cu^3_x+(pc+q)u^3_y+(a-b)u^1u^2.\label{eq:3W-sol}
\end{gather}
As prescribed by the general theory, the compatibility of the system \eqref{3-waves} with the Hamiltonian structure \eqref{3-W-S} is only a necessary condition for \eqref{3-waves} to be a Hamiltonian system. In fact, \eqref{3-waves} is Hamiltonian with a (1+0)-order operator if and only if $q=0$
and $P$ as in \eqref{3-W-S}, with the corresponding Hamiltonian functional
\begin{gather*}
H=\int \frac{a}{2S}{\big(u^1\big)}^2+\frac{b}{2S}{\big(u^2\big)}^2+\frac{c}{2S}{\big(u^3\big)}^2,
\end{gather*}
where $S$ is an arbitrary constant. Note that the Hamiltonian case is essentially one-dimensional, as it can be seen by the change of spatial variables $\tilde{x}=x$, $\tilde{y}=p x -y$ from which $\partial x+p\partial_y\mapsto\partial_{\tilde{x}}$.

\subsubsection{Relation with the Poisson cohomology}
Note that in this example the results of Proposition \ref{prop:biham} cannot be applied, and we cannot directly relate the Poisson cohomology of $P$ with the bi-Hamiltonian cohomology of the two homogeneous Poisson structures. Indeed, for $N=3$ we cannot rely either on the isomorphism between $H(P)$ and $H_{\rm dR}(M)$ established by Lichnerowicz \cite{lichnerowicz1977varietes}, because $P_0$ is not invertible, or on the triviality result for the first and second \emph{smooth} Poisson cohomology group of $\mathfrak{so}^*(3)$ proved by Conn \cite{conn85} and Ginzburg--Weinstein \cite{gw92}, because we consider a broader class of vector fields: indeed, we can explicitly obtain non-Hamiltonian symmetries of $P_0$, invalidating the assumption~${X_0={\rm d}_0h_0}$. Moreover, a linear change of independent variables as the one discussed in \cite{ccs17} simplifies the computations for the first-order operator and quickly leads to $H^1_1(P_1)\cong\mathbb{R}^6$.

However, the results of the previous paragraph show that there only exist a one-real parameter family of real 3-waves systems compatible with a Hamiltonian operators but not Hamiltonian -- indeed, we can rewrite \eqref{eq:3W-sol} as
\begin{gather*}
\begin{pmatrix}u^1_t\\[1mm] u^2_t\\[1mm] u^3_t\end{pmatrix}=\begin{pmatrix}a\big(u^1_x+p u^1_y\big)+(b-c)u_2u_3\\[1mm] b\big(u^2_x+p u^2_y\big)+(c-a)u_3u_2\\[1mm] c\big(u^3_x+p u^3_y\big)+(a-b)u_1u_2\end{pmatrix}+q\begin{pmatrix} u^1_y\\[1mm] u^2_y\\[1mm] u^3_y\end{pmatrix}
\end{gather*}
and the system is not Hamiltonian for $q\neq0$. It can be observed that the parameter $q$ can be set to 0 by a \emph{time-dependent} change of independent variables, $\tilde{x}=x$, $\tilde{y}=-y-qt$, $\tilde{t}=t$, giving~${\partial_t-q\partial_y\mapsto\partial_{\tilde{t}}}$.

The transformation between non-Hamiltonian ($q\neq 0$) and Hamiltonian ($q=0$) two-dimen\-sion\-al 3-waves equation is not in the class of natural transformations of evolutionary systems, as we already discussed in the previous example. However, Hamiltonian two-dimensional 3-waves equations are essentially monodimensional, as explained in the previous paragraph.

\section{Conclusion}
In this paper, the necessary conditions for a multidimensional first-order quasilinear system to admit a Hamiltonian structure are discussed. Firstly, we computed the conditions under which the (1+0)-order nonhomogeneous hydrodynamic operator satisfies the property of skewadjointness and $[P,P]=0$ using the PVA theory. Next, leveraging the theory of differential coverings, we have computed the compatibility conditions of multidimensional first-order quasilinear systems which admit (1+0)-order operators as their Hamiltonian structures. Finally, two examples for the $N=D=2$ system and a real reduction of the 3-waves system are presented. For the first case, we consider three Hamiltonian structures of hydrodynamic type along with the ultralocal structure. Based on the previous conclusions, we investigate which pairs of Hamiltonian structures possess Hamiltonian properties and search for compatible systems with such structures. The appearance of certain arbitrary constants in the compatible quasilinear systems indicates the existence of non-Hamiltonian systems. This consistency aligns with the results of the first Poisson cohomology group and highlights the essentially cohomological nature of the compatibility conditions. For the second case, we have obtained the Hamiltonian operator compatible with the 3-waves system and the constraints on the parameters in the system by computing compatibility conditions. Under specific conditions, we have derived the Hamiltonian structure and Hamiltonian functional of the 3-waves system.

There are natural extensions of this research. On the one hand, multi-dimensional \emph{nonlocal} Hamiltonian structures and equations, and their corresponding cohomological aspects, have not been deeply investigated yet -- works like Ferapontov and Mokhov's \cite{mokhov1994hamiltonian} deal with nonhomogeneous nonlocal structures but in one spatial dimension only. A future research direction is extending our study to equations admitting (weakly) multidimensional non-local structures. On the other hand, we showed how the Poisson cohomology of the nonhomogeneous Hamiltonian structures can be related to the bi-Hamiltonian cohomology of the homogeneous components, but in practice we performed explicit computations to obtain our result; computing bi-Hamiltonian cohomologies in general is very complicated, and until today it has been done only for homogeneous, first-order structures \cite{cps16,lz05}. Another future research direction is the study of bi-Hamiltonian cohomology in the nonhomogeneous case.

\subsection*{Acknowledgements}
This work was sponsored by the National Science Foundation of China (grant no.~12101341), Ningbo City Yongjiang Innovative Talent Program and Ningbo University Talent Introduction and Resarch Initiation Fund. M.C.\ wishes to thank P.~Vergallo and I.~Marcut for the useful insights.

\pdfbookmark[1]{References}{ref}
\LastPageEnding


\begin{thebibliography}{99}
\footnotesize\itemsep=0pt

\bibitem{ablowitz1975nonlinear}
Ablowitz M.J., Haberman R., Nonlinear evolution equations~-- two and three
 dimensions, \href{https://doi.org/10.1103/PhysRevLett.35.1185}{\textit{Phys. Rev. Lett.}} \textbf{35} (1975), 1185--1188.

\bibitem{barakat2009poisson}
Barakat A., De~Sole A., Kac V.G., Poisson vertex algebras in the theory of
 {H}amiltonian equations, \href{https://doi.org/10.1007/s11537-009-0932-y}{\textit{Jpn.~J. Math.}} \textbf{4} (2009), 141--252,
 \href{https://arxiv.org/abs/0907.1275}{arXiv:0907.1275}.

\bibitem{ccs17}
Carlet G., Casati M., Shadrin S., Poisson cohomology of scalar multidimensional
 {D}ubrovin--{N}ovikov brackets, \href{https://doi.org/10.1016/j.geomphys.2016.12.008}{\textit{J.~Geom. Phys.}} \textbf{114} (2017),
 404--419, \href{https://arxiv.org/abs/1512.05744}{arXiv:1512.05744}.

\bibitem{ccs18}
Carlet G., Casati M., Shadrin S., Normal forms of dispersive scalar {P}oisson
 brackets with two independent variables, \href{https://doi.org/10.1007/s11005-018-1076-x}{\textit{Lett. Math. Phys.}}
 \textbf{108} (2018), 2229--2253, \href{https://arxiv.org/abs/1707.03703}{arXiv:1707.03703}.

\bibitem{cps16}
Carlet G., Posthuma H., Shadrin S., Bihamiltonian cohomology of {K}d{V}
 brackets, \href{https://doi.org/10.1007/s00220-015-2540-4}{\textit{Comm. Math. Phys.}} \textbf{341} (2016), 805--819,
 \href{https://arxiv.org/abs/1406.5595}{arXiv:1406.5595}.

\bibitem{casati2015deformations}
Casati M., On deformations of multidimensional {P}oisson brackets of
 hydrodynamic type, \href{https://doi.org/10.1007/s00220-014-2219-2}{\textit{Comm. Math. Phys.}} \textbf{335} (2015), 851--894,
 \href{https://arxiv.org/abs/1312.1878}{arXiv:1312.1878}.

\bibitem{conn85}
Conn J.F., Normal forms for smooth {P}oisson structures, \href{https://doi.org/10.2307/1971210}{\textit{Ann. of Math.}}
 \textbf{121} (1985), 565--593.

\bibitem{dubrovin2001integrable}
Dubrovin B.A., Krichever I.M., Novikov S.P., Integrable systems.~{I}, in
 Dynamical {S}ystems,~{IV}, \textit{Encyclopaedia Math. Sci.}, Vol.~4,
 \href{https://doi.org/10.1007/978-3-662-06791-8_3}{Springer}, Berlin, 2001, 177--332.

\bibitem{dubrovin1996hamiltonian}
Dubrovin B.A., Novikov S.P., Hamiltonian formalism of one-dimensional systems
 of the hydrodynamic type and the {B}ogolyubov--{W}hitham averaging method,
 \textit{Dokl. Akad. Nauk SSSR} \textbf{270} (1983), 781--785.

\bibitem{dubrovin2001normal}
Dubrovin B.A., Zhang Y., Normal forms of hierarchies of integrable {PDE}s,
 {F}robenius manifolds and {G}romov--{W}itten invariants,
 \href{https://arxiv.org/abs/math.DG/0108160}{arXiv:math.DG/0108160}.

\bibitem{ferapontov2011classification}
Ferapontov E.V., Odesskii A.V., Stoilov N.M., Classification of integrable
 two-component {H}amiltonian systems of hydrodynamic type in {$2+1$}
 dimensions, \href{https://doi.org/10.1063/1.3602081}{\textit{J.~Math. Phys.}} \textbf{52} (2011), 073505, 28~pages,
 \href{https://arxiv.org/abs/1007.3782}{arXiv:1007.3782}.

\bibitem{getzler}
Getzler E., A {D}arboux theorem for {H}amiltonian operators in the formal
 calculus of variations, \href{https://doi.org/10.1215/S0012-7094-02-11136-3}{\textit{Duke Math.~J.}} \textbf{111} (2002), 535--560,
 \href{https://arxiv.org/abs/math/0002164}{arXiv:math/0002164}.

\bibitem{gw92}
Ginzburg V.L., Weinstein A., Lie--{P}oisson structure on some {P}oisson {L}ie
 groups, \href{https://doi.org/10.2307/2152773}{\textit{J.~Amer. Math. Soc.}} \textbf{5} (1992), 445--453.

\bibitem{kersten2004hamiltonian}
Kersten P., Krasil'shchik I., Verbovetsky A., Hamiltonian operators and
 {$l^*$}-coverings, \href{https://doi.org/10.1016/j.geomphys.2003.09.010}{\textit{J.~Geom. Phys.}} \textbf{50} (2004), 273--302,
 \href{https://arxiv.org/abs/math/0304245}{arXiv:math/0304245}.

\bibitem{krasil1989nonlocal}
Krasil'shchik I.S., Vinogradov A.M., Nonlocal trends in the geometry of
 differential equations: symmetries, conservation laws, and {B}\"acklund
 transformations: Symmetries of partial differential equations, Part~{I},
 \href{https://doi.org/10.1007/BF00131935}{\textit{Acta Appl. Math.}} \textbf{15} (1989), 161--209.

\bibitem{lichnerowicz1977varietes}
Lichnerowicz A., Les vari\'et\'es de {P}oisson et leurs alg\`ebres de {L}ie
 associ\'ees, \href{https://doi.org/10.4310/jdg/1214433987}{\textit{J.~Differential Geometry}} \textbf{12} (1977), 253--300.

\bibitem{lz05}
Liu S.-Q., Zhang Y., Deformations of semisimple bihamiltonian structures of
 hydrodynamic type, \href{https://doi.org/10.1016/j.geomphys.2004.11.003}{\textit{J.~Geom. Phys.}} \textbf{54} (2005), 427--453,
 \href{https://arxiv.org/abs/math.DG/0405146}{arXiv:math.DG/0405146}.

\bibitem{mokhov1988poisson}
Mokhov O.I., Poisson brackets of {D}ubrovin--{N}ovikov type ({DN}-brackets),
 \href{https://doi.org/10.1007/BF01077434}{\textit{Funct. Anal. Appl.}} \textbf{22} (1988), 336--338.

\bibitem{mokhov1995symplectic}
Mokhov O.I., Symplectic and {P}oisson structures on loop spaces of smooth
 manifolds, and integrable systems, \href{https://doi.org/10.1070/RM1998v053n03ABEH000019}{\textit{Russian Math. Surveys}} \textbf{53}
 (1998), 515--622.

\bibitem{mokhov2008classification}
Mokhov O.I., The classification of nonsingular multidimensional
 {D}ubrovin--{N}ovikov brackets, \href{https://doi.org/10.1007/s10688-008-0004-8}{\textit{Funct. Anal. Appl.}} \textbf{42}
 (2008), 33--44, \href{https://arxiv.org/abs/math.DG/0611785}{arXiv:math.DG/0611785}.

\bibitem{mokhov1994hamiltonian}
Mokhov O.I., Ferapontov E.V., Hamiltonian pairs generated by skew-symmetric
 {K}illing tensors on spaces of constant curvature, \href{https://doi.org/10.1007/BF01076502}{\textit{Funct. Anal.
 Appl.}} \textbf{28} (1994), 123--125.

\bibitem{olver1980hamiltonian}
Olver P.J., On the {H}amiltonian structure of evolution equations,
 \href{https://doi.org/10.1017/S0305004100057364}{\textit{Math. Proc. Cambridge Philos. Soc.}} \textbf{88} (1980), 71--88.

\bibitem{olver1993applications}
Olver P.J., Applications of {L}ie groups to differential equations, 2nd ed.,
 \textit{Grad. Texts Math.}, Vol.~107, \href{https://doi.org/10.1007/978-1-4612-4350-2}{Springer}, New York, 1993.

\bibitem{tsarev1989hamilton}
Tsar\"ev S.P., The {H}amilton property of stationary and inverted equations in
 continuum mechanics and mathematical physics, \href{https://doi.org/10.1007/BF01159109}{\textit{Math. Notes}}
 \textbf{46} (1989), 569--573.

\bibitem{Vergallo_2022}
Vergallo P., Quasilinear systems of first order {PDE}s with nonlocal
 {H}amiltonian structures, \href{https://doi.org/10.1007/s11040-022-09438-1}{\textit{Math. Phys. Anal. Geom.}} \textbf{25}
 (2022), 26, 19~pages, \href{https://arxiv.org/abs/2203.13554}{arXiv:2203.13554}.

\bibitem{vergallo2023non}
Vergallo P., Non-homogeneous {H}amiltonian structures for quasilinear systems,
 \href{https://doi.org/10.1007/s40574-023-00369-5}{\textit{Boll. Unione Mat. Ital.}} \textbf{17} (2024), 513--526,
 \href{https://arxiv.org/abs/2301.10713}{arXiv:2301.10713}.

\bibitem{vergallo2021homogeneous}
Vergallo P., Vitolo R., Homogeneous {H}amiltonian operators and the theory of
 coverings, \href{https://doi.org/10.1016/j.difgeo.2020.101713}{\textit{Differential Geom. Appl.}} \textbf{75} (2021), 101713,
 16~pages, \href{https://arxiv.org/abs/2007.15294}{arXiv:2007.15294}.

\bibitem{weinstein1983local}
Weinstein A., The local structure of {P}oisson manifolds,
 \href{https://doi.org/10.4310/jdg/1214437787}{\textit{J.~Differential Geom.}} \textbf{18} (1983), 523--557.

\end{thebibliography}
\end{document}